%% file: deterministic_transshipment.tex
\documentclass[11pt,twoside]{article}
\usepackage[utf8]{inputenc}
\usepackage{jeffe}
\usepackage[hmargin=1in,vmargin=1in]{geometry}
\usepackage{cite}
\usepackage{cleveref}
\usepackage{xspace}
\usepackage{enumitem}
\usepackage{amsopn}
\usepackage{verbatim}

\usepackage[mathlines]{lineno}

\DeclareMathOperator{\dist}{dist}

\def\eps{\varepsilon}
\def\demand{b}
\def\cost{c}

\def\dartof#1{\vec{#1}}
\def\dartsof#1{\dartof{#1}}
\def\dist{\delta}

\def\reach{\Delta}
\def\numlayers{L}
\def\contractset{V}
\def\parent{p}
\def\children{\parent^{-1}}
\def\ancestors{\parent^{\infty}}

\def\lcapprox{P}
\def\lcaggregate{A}
\def\lcdistribute{D}
\def\lcroute{R}
\def\lccost{C}
\def\approxratio{\alpha}
\def\opt{\textsc{OPT}}
\def\shortestpath{\pi}
\def\representative{r}
\def\incidence{I}
\def\close{N}

\newtheorem{lemma}{Lemma}[section]
\newtheorem{theorem}[lemma]{Theorem}

\begin{document}

\begin{titlepage}
  \title{A simple deterministic near-linear time approximation scheme for transshipment with arbitrary positive edge costs}

\author{
Emily Fox%
\thanks{Department of Computer Science,
 The University of Texas at Dallas; \url{emily.fox@utdallas.edu}.
 Supported in part by NSF grant CCF-1942597.}
}


\maketitle

\begin{abstract}
  We describe a simple deterministic near-linear time approximation scheme for uncapacitated minimum cost flow in undirected graphs with positive real edge weights, a problem also known as transshipment.
  Specifically, our algorithm takes as input a (connected) undirected graph \(G = (V, E)\), vertex demands \(\demand \in \R^V\) such that \(\sum_{v \in V} \demand(v) = 0\), positive edge costs \(\cost \in \R_{>0}^E\), and a parameter \(\eps > 0\).
  In \(O(\eps^{-2} m \log^{O(1)} n)\) time, it returns a flow \(f\) such that the net flow out of each vertex is equal to the vertex's demand and the cost of the flow is within a \((1 + \eps)\) factor of optimal.
  Our algorithm is combinatorial and has no running time dependency on the demands or edge costs.

  With the exception of a recent result presented at STOC 2022 for polynomially bounded edge weights, all almost- and near-linear time approximation schemes for transshipment relied on randomization to embed the problem instance into low-dimensional space.
  Our algorithm instead deterministically approximates the cost of routing decisions that would be made if the input were subject to a random \emph{tree embedding}.
  To avoid computing the \(\Omega(n^2)\) vertex-vertex distances that an approximation of this kind suggests, we also take advantage of the clustering method used in the well-known Thorup-Zwick distance oracle.
\end{abstract}

\setcounter{page}{0}
\thispagestyle{empty}
\end{titlepage}

\pagestyle{myheadings}
\markboth{Emily Fox}
{A simple deterministic near-linear time approximation scheme for transshipment}

\input{introduction}
\input{oracle}
\input{layers}
\input{linear_cost_approx}
\input{approx_analysis}

\input{scheme}

\bibliographystyle{alpha} 
\bibliography{refs}

\end{document}

%% file: introduction.tex
\section{Introduction}
\label{sec:introduction}

Let \(G = (V, E)\) be an undirected graph with positive edge \EMPH{costs} \(\cost \in \R_{>0}^E\), and let \(\demand \in \R^V\) be a set of vertex \EMPH{demands} (alternatively, one may prefer the term \emph{supplies}).
While we formally define \(\cost\) and \(\demand\) as vectors with components indexed by \(E\) and \(V\), respectively, we use the familiar function application notation \(\cost(e)\) and \(\demand(v)\) for the cost of an edge \(e \in E\) and demand for a vertex \(v \in V\), respectively.
We say \(\demand\) is \EMPH{proper} if \(\sum_{v \in V} \demand(v) = 0\).

Let \(\dartsof{E}\) denote an arbitrary orientation of the edges \(E\).
We denote the oriented instance of an edge \(e \in E\) as \(\dartsof{e}\).
Let \(\incidence_G \in \R^{V \times \dartsof{E}}\) be the vertex-edge incidence matrix for \(G\) with \(\incidence_G(v, \dartsof{e})\) equal to \(1\) if \(v\) is the tail of \(\dartsof{e}\), equal to \(-1\) if \(v\) is the head of \(\dartsof{e}\), and equal to \(0\) otherwise.
We say a \EMPH{flow} \(f \in \R^{\dartsof{E}}\) \EMPH{routes} \(\demand\) if \(\incidence_G f = \demand\).
In the (uncapacitated) minimum cost flow problem, one seeks a flow of minimum cost \(\cost(f) = \sum_{e \in E} \cost(e) |f(\dartsof{e})|\) subject to \(f\) routing \(\demand\).
In other words, we seek a minimum cost way to send units of some single commodity throughout the edges of \(G\) such that each vertex \(u \in V\) with \(\demand(u) > 0\) sends out \(\demand(u)\) units of commodity into the graph and each vertex \(v \in V\) with \(\demand(v) < 0\) removes \(-\demand(v)\) units from the graph.
This special case of minimum cost flow in an undirected graph without edge capacities is also called \EMPH{transshipment}.

Transshipment generalizes various problems that have been studied in their own right such as shortest paths in undirected graphs, the discrete optimal transport problem~\cite{v-oton-08}, and other assignment problems on metric spaces.
In fact, several recent papers studying these more specific problems have relied on reductions to the more general transshipment problem~\cite{knp-pgtp-21,fl-ntasg-22,bfkl-naspt-21,asz-pausp-20,l-fpaas-20,zgyhs-udspt-22,rghzl-u1pmn-22,fl-dntas-23}.
The study of new algorithms for transshipment can provide immediate improvements or simplifications to many of these results along with providing new insights that may be beneficial to minimum cost and other flow problems in general.

\subsection{Recent results}
\label{sec:introduction-recent}

The study of flow problems such as transshipment has a long history going back several decades.
Here, we highlight some of the strongest or more recent closely related results to the current work.
We use \(n\) and \(m\) to denote the number of vertices and edges, respectively, in the input graph.

As a special case of minimum cost flow, there are several polynomial time algorithms for computing exact solutions to transshipment.
Orlin~\cite{o-fspmc-93} described a strongly polynomial transshipment algorithm that runs in \(O(n \log n(m + n \log n))\) time, and this algorithm remains the fastest algorithm known for real edge costs and vertex demands.
There has been a great deal of recent activity in the design of minimum cost flow and transshipment algorithms that assume integer costs and capacities or demands in some range~\([1, U]\), starting with an \(O(m^{3/2} \log^{O(1)} (nU))\) time algorithm by Daitch and Spielman~\cite{ds-falgf-08} and culminating in a pair of very recent \EMPH{almost-linear} \(m^{1 + o(1)} \log^2 U\) time algorithms~\cite{cklpgs-mfmfa-22,bcpklgss-datam-23}.

The existence of almost-linear time exact algorithms for minimum cost flow was alluded to a few years earlier by the demonstration of various almost- and near-linear time \EMPH{approximation schemes} for transshipment.
Let \(\eps > 0\), and let \(\opt(\demand)\) denote the minimum cost of any flow that routes \(\demand\).
Sherman~\cite{s-gpumf-17} described an \(O(\eps^{-2} m^{1 + o(1)})\) time algorithm that finds a flow \(f\) routing \(\demand\) with total cost \(\cost(f) \leq (1 + \eps)\opt(\demand)\).
Sherman's main observation was a novel method for finding solutions to linear systems \(Ax = d\) that approximately minimize an arbitrary norm \(||x||\).
His method involved composing solutions from well-known weak approximate solvers by repeatedly applying the solver to residual vectors \(d - Ax\).
The number of iterations needed for this method to converge is a function of the so-called \EMPH{generalized condition number} of \(A\).
To reduce the condition number, he proposed finding a left-cancellable matrix~\(P\) called a \EMPH{generalized preconditioner} such that \(PA\) is well-conditioned and then working with the system \(PAx = Pb\).
He then expressed transshipment as such a linear system problem and described a preconditioner \(P\) that could be efficiently applied in iterations of his composition algorithm.

Recently, the authors of \cite{l-fpaas-20} and \cite{asz-pausp-20} independently discovered \EMPH{near-linear}\linebreak \(O(\eps^{-2} m \log^{O(1)} (nU))\) time approximation schemes for transshipment.
Here, \(U\) is best understood as the \EMPH{aspect ratio} of the edge costs found by dividing the largest edge cost by the smallest.
Shortly after, Fox and Lu~\cite{fl-dntas-23} proposed a near-linear \(O(\eps^{-2} m \log^{O(1)} n)\) time approximation scheme \emph{without} the dependence on the aspect ratio.
While the above results were presented here with sequential running times in mind, there have been several recent approximation schemes proposed for various models of parallel and distributed computing, including some appearing in a subset of the work cited above~\cite{l-fpaas-20,asz-pausp-20,bfkl-naspt-21,rghzl-u1pmn-22,zgyhs-udspt-22}.

Perhaps unsurprisingly, the algorithms of~\cite{asz-pausp-20} and~\cite{fl-dntas-23} use the aforementioned framework of Sherman~\cite{s-gpumf-17} explicitly but with a more-efficiently evaluated choice for the preconditioner~\(P\).
However, \emph{all} of the approximation schemes for transshipment mentioned above rely on methods for refining loose approximate solutions into stronger ones.
Zuzic~\cite{z-sbft-23} recently provided an explanation for this commonality by uniting the approaches of these works under a single simple \EMPH{boosting} framework.
In short, all of these works implicitly build approximately optimal results to the linear programming dual for transshipment and then take advantage of the newly realized fact that \emph{any} black-box dual approximation can be boosted to a \((1 + \eps)\)-approximate solution for transshipment.

Another commonality between all of the almost- and near-linear time approximation schemes cited above, with a single exception, is that they rely \emph{heavily} on randomization.
In particular, they all compute random Bourgain~\cite{b-lefms-85} embeddings of the shortest path metric into low-dimensional space and then they compute (the cost of) random \emph{oblivious} flows that are based only on the location of their sources within that space.
The single near-linear time exception to this use of randomization is a recent paper~\cite{rghzl-u1pmn-22} describing an~\(O(\eps^{-2} m \log^{O(1)} (nU))\) time \emph{deterministic} approximation scheme.
However, unlike some of the results mentioned above~\cite{s-gpumf-17,fl-dntas-23}, its running time is still polylogarithmic in the aspect ratio of the edge costs.

\subsection{Our results}
\label{sec:introduction-results}

We present the first near-linear time approximation scheme for transshipment that is both deterministic and with a running time independent of the aspect ratio of the edge costs.
Specifically, our algorithm computes a flow \(f\) that routes \(\demand\) at cost \(\cost(f) \leq (1 + \eps)\opt(\demand)\) in \(O(\eps^{-2}m \log^{O(1)} n)\) time.
It is also (in our opinion) simpler and likely easier to implement than previous near-linear time approximation schemes for transshipment, even after excising the extra complications needed for them to efficiently function in parallel or distributed settings.
Outside what is explicitly described in this report, it depends upon just two black box results, the aforementioned boosting framework of Zuzic~\cite{z-sbft-23} implicitely used by all previous almost- and near-linear time transshipment approximation schemes and the deterministic construction~\cite{rtz-dcado-05} of a well-known distance oracle of Thorup and Zwick~\cite{tz-ado-05}.
Our algorithm is also combinatorial in that the only operations it performs with the input costs and demands are comparisons, addition, multiplication, and division.

\subsubsection*{Linear cost approximators}
\paragraph*{Prior work}

One method of instantiating the transshipment boosting framework~\cite{z-sbft-23} is to design an \EMPH{\(\approxratio\)-approximate linear cost approximator} \(\lcapprox \in \R^{k \times V}\) based only on the input graph \(G = (V, E)\) and edge costs \(\cost\).
In particular, for \emph{any} set of proper demands \(\demand\), we must have \(\opt(\demand) \leq ||\lcapprox \demand||_1 \leq \approxratio \opt(\demand)\).
Approximator \(\lcapprox\) need not be computed explicitly.
If matrix-vector multiplications with \(\lcapprox\) and \(\lcapprox^T\) can be performed in some time \(M\), then we can compute a flow \(f'\) with \(\cost(f') \leq (1 + \eps / 2)\opt(\demand)\) and \(\opt(\demand - \incidence_G f') \leq \opt(\demand) / n^2\) in \(O(\eps^{-2}\alpha^2 M \log^{O(1)} n)\) time~\cite[Corollaries 12 and 16]{z-sbft-23}.
We can then route an \(n\)-approximate flow for demands \(\demand - \incidence_G f'\) along a minimum spanning tree to get our desired \((1 + \eps)\)-approximate flow that routes \(\demand\) exactly.

Linear cost approximators are almost the same as the generalized preconditioners mentioned above, and we can look to prior work on how to design one.
In particular, the construction we use is partially motivated by the near-linear time approximation schemes of~\cite{asz-pausp-20,fl-dntas-23}.
As in most of the previous approximation schemes for transshipment, they compute a random Bourgain~\cite{b-lefms-85} embedding of the input graph's shortest path metric, mapping vertices to points in low dimensional space.
They then construct preconditioners for estimating the optimal solution value to the \emph{geometric transportation problem} over the vertices' points.
In this latter problem, the goal is to compute a weighted matching between several pairs of points of minimum total distance.

Consider a hierarchy of subsets of the vertices' points based on a sequence of progressively finer randomly shifted uniform grids.
By greedily matching points within lower levels of the hierarchy before moving to the top, one obtains a weighted matching with expected cost close to optimal.
The approximation schemes of~\cite{asz-pausp-20,fl-dntas-23} construct preconditioners to estimate this expected cost by building a collection of \emph{deterministic} (i.e., not randomly shifted) grids and explicitly computing the net expected amount of demand within each grid cell, \emph{as if they had been randomly shifted}.
Their preconditioners (modulo appropriate scaling) simply output the diameter of each grid cell times the net expected demand it contains.

\paragraph*{Novel construction inspired by tree embeddings}
Instead of using a random Bourgain embedding and then building a hierachy of subsets, our approximation scheme skips straight to considering the hierarchies formed from random embeddings into dominating \emph{tree metrics}~\cite{frt-tbaam-04} where the distance between any given pair of vertices is stretched (distorted) by a factor of at most \(O(\log n)\) in expectation.
Consider the following variation of the tree embedding of~\cite{frt-tbaam-04}.
Let \(r\) be the root of our tree.
We compute a \(2\)-approximation \(\reach\) of the diameter of \(G\) and choose a partition \(\Seq{v_1, v_2, \dots, v_n}\) of the vertices uniformly at random.
We create a sequence of disjoint clusters \(\Seq{C_1, C_2, \dots, C_n}\) where \(C_i \subseteq V\) for all \(i\).
Let \(\reach'\) be chosen uniformly at random from \([\reach / 2, \reach]\).
For each \(i\) from \(1\) to \(n\), we add to \(C_i\) all vertices within distance \(\reach'\) that have not already been claimed for another cluster.
We create a child \(c_i\) of \(r\) for each non-empty cluster \(C_i\), connected by an edge of weight \(O(\Delta)\).
Finally, we recursively build a tree rooted at \(c_i\) for each non-empty \(C_i\).
Despite the low expected stretch, some distances may be distorted by a factor of \(\Omega(n)\).
So while for any fixed \(\demand\), the expected cost of an optimal flow routing it in the tree is \(O(\log n) \cdot \opt(\demand)\), there may be some choices of \(\demand\) for which the cost blows up by that \(\Omega(n)\) factor.
Therefore, we cannot simply create a linear cost approximator based on the cost of routing different demands within the tree and expect it to give us good approximate costs relative to \(G\) for the large number of \(\demand\) vectors used in the transshipment boosting framework.

However, the \emph{expected} costs of routing demand through potential cluster centers is a fixed value that can be computed accurately.
Ignoring running time concerns, we can construct a linear cost approximator \(\lcapprox\) that accurately lists these expected costs up to constant factors.
Because the expected stretch from an actual random tree embedding using the above algorithm is \(O(\log n)\), the value \(||\lcapprox \demand||_1 \leq O(\log n) \cdot \cost(\demand)\) always.

\subsubsection*{Distance oracles and graph minors}
Unfortunately, we cannot afford to compute the \(\Omega(n^2)\) distances required to properly compute these expected costs.
Instead we take advantage of the clustering method used
in the well-known distance oracle of Thorup and Zwick~\cite{tz-ado-05}.
Let \(k \geq 1\) be an integer.
After \(O(k m n^{1/k} \log n)\) time deterministic preprocessing, their \(O(kn^{1 + 1/k})\)-space oracle can compute \((2k - 1)\)-approximate distances between any pair of vertices in \(O(k)\) time~\cite{rtz-dcado-05,tz-ado-05}.
By setting \(k := \lg n\), we get an \(O(\log n)\)-approximate distance oracle with construction time \(O(m \log^2 n)\) and size \(O(n \log n)\).
We never directly use the construction for its stated purpose as an oracle.
Instead, we use the fact that the deterministic construction stores for each \(v \in V\) a set of \(O(\log^2 n)\) vertices \(w \in V\) that suffice as the possible cluster centers for our expected cost computations.
The actual algebra proving we get a good approximator using the expected costs is, unsurprisingly, very similar to the algebra proving a random tree embedding has low expected stretch.

There is one thing left to consider in the design of our algorithm.
Even using the distance oracle, we would have to consider all possible distance scales to get correct expected costs for all possible cluster centers across all possible cluster diameters.
Doing so would lead to a polylogarithmic dependence on the aspect ratio in our running time.
To avoid that dependence, we construct a sequence of \(O(m \log n)\) minors of \(G\), each maintaining a range of possible shortest path distances up to a constant factor.
Edges of higher cost than the range of a minor are discarded, and those of significantly smaller cost are contracted, so the total size of these minors is also \(O(m \log n)\).
Our construction of the linear cost approximator \(\lcapprox\) simply considers expected costs within each minor separately.
When establishing the approximation ratio of \(\lcapprox\), we charge against the cost of individual flow paths in a decomposition of an optimal flow.
Fortunately, only \(O(\log n)\) minors charge meaningful costs to each flow path, bringing the approximation ratio of \(\lcapprox\) to a relatively small \(O(\log^3 n)\).

\subsubsection*{Oblivious routing and comparison to \cite{rghzl-u1pmn-22}}
Along with the explanation given above, the linear cost approximator \(\lcapprox\) can be interpreted as providing constant approximations to the actual cost of a certain \EMPH{oblivious} flow where a unit flow from/to each vertex \(u\) to/from an arbitrary vertex \(s\) is chosen without prior knowledge of \(\demand\) and then multiplied by \(\demand(u)\).
The description of the flow is incredibly simple;
for each adjacent pair of scales for which we consider how a random tree embedding might affect~\(u\), between each pair of potential cluster centers between those scales, the unit flow for \(u\) sends the product of the probabilities that \(u\) would join their respective clusters.
Another way to think about the flow is that for each scale, we want the demand of \(u\) to arrive at various nearby cluster centers, and using the product of proportions between scales to route flow is the most natural way to do so.
The actual construction of~\(\lcapprox\) follows this oblivious routing/reassignment interpretation, as we believe it more easily suggests that~\(\lcapprox\) is estimating the cost of an actual solution to the transshipment problem.
That said, our algorithm never actually computes an oblivous flow, because we merely need applications of its cost approximator \(\lcapprox\) to use the boosting framework.
Also, the actual choice of \emph{where} to send the flow units and the analysis for~\(\lcapprox\)'s approximation ratio is much better explained using the tree embedding motivation described above.

This oblivious routing of actual flow interpretation/approach is used much more heavily throughout the \(O(\eps^{-2}m \log^{O(1)}(nU))\) time deterministic approximation scheme of \cite{rghzl-u1pmn-22}, so it provides a means by which to compare our work to theirs.
Their paper contains the many additional details necessary for working in parallel and distributed models that are beyond the scope of the current work, so we focus just on the parts related to approximating transshipment in any model.
Similar to how our approach is inspired by random tree embeddings, theirs is inspired by random \emph{low-diameter decompositions} of the input graph at different scales and their deterministic counterparts, the \emph{sparse neighborhood covers}.
They observe that sending of a vertex \(u\)'s demand from cluster center to cluster center can lead to an oblivious flow having high cost, so they propose sending portions of the demand to nearby cluster centers proportionally to their distance and then routing flow between scales based on the product of these proportions.
Our expected cost/demand reassignment calculations also bias sending demand to nearby cluster centers.
However, the algorithm of \cite{rghzl-u1pmn-22} and its analysis is made more complicated compared to ours by, for example, them only sending flow between centers of nesting clusters.

Also similar to our work, they build various simplifications of the input graph designed to efficiently consider clusters of a certain scale.
They build \(O(\log (nU))\) simplifications to handle all possible different scales they might need to consider.
The analysis of their oblivious routing's cost must charge to the full cost of the optimal flow paths once per scale, leading to their obliviously routed flow having a cost \(O(\log^{O(1)}(nU))\) times optimal.
This approximation ratio then becomes part of their algorithm's running time as in all boosting based approximation schemes.
In contrast, our minors have total size \(O(m \log n)\) and are designed so each path in an optimal flow will receive significant charges from only \(O(\log n)\) of them, leading to a tidy \(O(\log^3 n)\) approximation ratio for our linear cost approximator.

\subsection{Organization}

We proceed as follows.
We discuss a few more needed details concerning the Thorup-Zwick distance oracle in Section~\ref{sec:oracle}.
We discuss the construction of the minor graphs (thereafter referred to as \emph{layers} of \(G\)) in Section~\ref{sec:layers}.
The construction and application of our \(O(\log^3 n)\)-approximate linear cost approximator \(\lcapprox\) is given in Section~\ref{sec:linear_cost_approx};
the reader merely interested in \emph{how} our algorithm works can stop there given the description of the boosting framework available above.
In Section~\ref{sec:approx_analysis}, we prove \(\lcapprox\) is an \(O(\log^3 n)\)-approximate linear cost approximator.
We briefly wrap things up in Section~\ref{sec:scheme} with the presentation of a theorem stating our main result.

%% file: oracle.tex
\section{Thorup-Zwick distance oracle}
\label{sec:oracle}

Thorup and Zwick~\cite{tz-ado-05} presented a \EMPH{distance oracle} that for any integer \(k \geq 1\) has size \(O(k n^{1 + 1/k})\) and can return \((2k - 1)\)-approximate distances between any two vertices in \(O(k)\) time.
It can be constructed deterministically in \(O(k m n^{1/k} \log n)\) time~\cite{rtz-dcado-05,tz-ado-05}.
We now discuss some more details of the oracle relevant to our algorithm.

Let \(\dist(u, v)\) denote the distance between vertices \(u\) to \(v\) in \(G = (V, E)\) and let \(\dist(u, V') = \min_{v \in V'} \dist(u, v)\) for any subset \(V' \subseteq V\).
The distance oracle stores a sequence of vertex subsets \(V = S^0 \supseteq S^1 \supseteq \dots \supseteq S^k = \emptyset\) we refer to as \EMPH{samples}.
We have \(S^{k - 1} \neq \emptyset\).
We assume distances between any fixed vertex \(v\) and the other vertices of \(G\) are distinct, breaking ties as necessary.
The oracle also stores, for each vertex \(v\), a \EMPH{bundle} \(B(v) = \cup_{j = 0}^{k - 1} B^j(v)\) of vertices and their distances from \(v\) where each \EMPH{bundle piece} \(B^j(v) = \Set{w \in S^j \mid \dist(v, w) < \dist(v, S^{j + 1})}\).
In particular, \(B^j(v) \subseteq S^j \setminus S^{j + 1}\).

The total size of all bundles is \(O(k n^{1 + 1/k})\).
Some bundles may have size larger than the average size of \(O(k n^{1/k})\).
However, a careful examination of the deterministic construction of the oracle~\cite{rtz-dcado-05} shows \(|B(v)| = O(k n^{1/k} \log n)\) for all \(v \in V\).

%% file: layers.tex
\section{Layer graphs}
\label{sec:layers}

Let \(G = (V, E)\) be a connected undirected graph with positive edge costs \(c \in \R_{>0}^E\), and let \(n := |V|\) and \(m := |E|\).
We assume without loss of generality that \(G\) contains no loops or parallel edges and that \(n \geq 4\).
Our approximation scheme begins by computing a sequence \(\langle (G_0, \reach_0), (G_1, \reach_1), \dots, (G_{\numlayers}, \reach_{\numlayers})\rangle\) of pairs, each consisting of a \emph{minor} \(G_i\) of \(G\), also referred to as a \EMPH{layer} of \(G\), and a \EMPH{reach} \(\reach_i\) such that \(\reach_{i} \leq \reach_{i-1} / 2\) for all \(i \geq 1\).
Each iteration of the approximate optimization procedure will take time near-linear in the sum of the minors' sizes, so we must make sure that each edge of \(G\) appears in only \(O(\log n)\) different minors.
Accordingly, our construction sets each \(G_i\), including for \(i = 0\), to be the graph \(G\) after contracting all edges of cost at most \(\reach_i / n\), deleting all edges of cost strictly greater than \(2\reach_i\), and removing all vertices left isolated given the edge deletions and contractions.
We let \(V_i\) and \(E_i\) denote the vertices and edges, respectively, of each layer \(G_i\), and let \(n_i := |V_i|\) and \(m_i := |E_i|\) denote the cardinality of both sets.
Let \(\dist_i(v, w)\) and \(\dist_i(v, W)\) denote the distance from \(v\) to another vertex or set of vertices within \(G_i\).

Let \(s\) be an arbitrary vertex of \(G\), and let \(x\) and \(y\) be the two vertices farthest from \(s\).
We set \(\reach_0 := \dist(s, x) + \dist(s, y)\).
Reach \(\reach_0\) is at least, but no more than twice, the diameter of \(G\).
Given \(\reach_{i-1}\), we set \(\reach_{i}\) as follows:
If \(G_{i-1}\) is non-empty (contains at least one edge), then \(\reach_i := \reach_{i - 1} / 2\).
Otherwise, let \(e^{||} = \argmax_{e \in E \mid \cost(e) \leq \reach_{i-1} / n} c(e)\) be the costliest contracted edge of \(G_{i-1}\).
If \(e^{||}\) is well-defined, let \(\reach_i := \cost(e^{||}) \cdot n / 2\).
If \(e^{||}\) is not well-defined, then \((G_{i - 1}, \reach_{i-1})\) is the final pair in the sequence and \(\numlayers := i - 1\).

For a vertex \(v\) in some layer \(G_i\), we let \(\contractset(v)\) denote the set of vertices from the input graph \(G\) contracted to form \(v\).
Observe that the sets \(\contractset(\cdot)\) form a \emph{laminar family} in that for each pair of sets, either they are disjoint or one completely contains the other.
Accordingly, let \(i'\) be the largest index such that \(i' < i\) and there exists a vertex \(v' \in V_{i'}\) such that \(V(v) \subseteq V(v')\).
We define the \EMPH{parent} of \(v\) to be \(\parent(v) := v'\).

The \EMPH{ancestors} of \(v \in V_i\), denoted \(\ancestors(v)\) are all layer graph vertices obtained by repeatedly applying the parent operation zero or more times starting with \(v\).
In particular, \(\parent_{i'}(v)\) for some \(i' \leq i\) denotes the ancestor of \(v\) in \(V_{i'}\) if one exists.
The \EMPH{children} of \(v\) are \(\children(v) := \Set{v' \mid \parent(v') = v}\).
Vertex \(v\) is called a \EMPH{leaf} if it has no children.
Despite the evocative names, we do not actually connect the layer graphs using any kind of rooted forest data structures and instead use them mostly separately when defining our linear cost approximator.

\begin{lemma}
\label{lem:layers_size}
    The graphs~\(\langle G_0, G_1, \dots, G_{\numlayers}\rangle\) have at most \(O(m \log n)\) edges and vertices in total.
\end{lemma}
\begin{proof}
    Let \(e \in E\) be any edge of \(G\), and let \(G_j\) be the first graph in the sequence to contain \(e\).
    Because \(e\) is not contracted in \(G_j\), we have \(\reach_j / n < \cost(e)\).
    As stated, \(\reach_i \leq \reach_{i - 1} / 2\) for all \(i \geq 1\).
    For all \(i \geq j + \lg n + 1\), we have \(\reach_i \leq \reach_j / (2n) < \cost(e)/2\) and \(e\) is deleted from \(G_i\).
    The edge \(e\) exists in at most \(1 + \lg n\) graphs in the sequence, implying the sequence of graphs contains \(O(m \log n)\) edges in total.
    The graphs contain no isolated vertices, so the total number of vertices is also \(O(m \log n)\).
\end{proof}

\begin{lemma}
\label{lem:construct_layers}
    The sequence~\(\langle (G_0, \reach_0), (G_1, \reach_1), \dots, (G_{\numlayers}, \reach_{\numlayers})\rangle\) along with their vertices' parents and children can be computed in \(O(m \log n)\) time.
\end{lemma}
\begin{proof}
    Reach \(\reach_0\) can be computed in \(O(m \log n)\) time by running Dijkstra's algorithm~\cite{d-ntpcg-59} with a binary heap to compute shortest paths.
    We store each edge of cost at most \(\reach_0 / n\) in a binary heap keyed by cost, contract these edges, and delete edges of cost more than \(2\reach_0\) in \(O(m)\) additional time.
    We now iteratively compute \((G_i, \reach_i)\) for each \(i \geq 1\).
    Suppose we've computed \((G_{i - 1}, \reach_{i-1})\).
    We compute \(\reach_i\) in \(O(1)\) time by peeking at the top of the binary heap if necessary.
    We repeatedly remove edges from the binary heap that should belong to \(G_i\) but not \(G_{i-1}\).
    We then uncontract these edges and delete costly edges from \(G_{i-1}\) to create \(G_i\) in \(O(m_{i-1} + m_i)\) time.
    We add parent pointers during this time as well.
    After all layers have been constructed, we loop over each vertex one more time to add a child pointer from their parent.
    In total, we spend \(O(m \log n)\) time removing edges from the binary heap and \(O(m \log n)\) time constructing the individual graphs (Lemma~\ref{lem:layers_size}).
\end{proof}

\begin{lemma}
\label{lem:parent_structure}
    Let \(v \in V_i\) for \(i >  0\), and let \(\parent(v) \in V_{i'}\) with \(\reach_{i'} > 2\reach_i\).
    The children of \(\parent(v)\) are exactly the members of the connected component of \(v\) in \(G_i\).
    Further, all edges incident to \(\parent(v)\) have length strictly greater than \(\reach_{i'}\), and the diameter in \(G\) of \(\contractset(\parent(v)) < 2 \reach_i\).
\end{lemma}
\begin{proof}
    All edges leaving \(\parent(v)\) are absent in layer \(i'\) or \(i' + 1\), so they have length strictly greater than \(2 \cdot (\reach_{i'} / 2) = \reach_{i'}\).
    Children of \(\parent(v)\) appear only when the reach is small enough to uncontract some of their edges.
    No other edges are added except those being uncontracted, so the entire connected component consists of children of \(\parent(v)\).
    Finally, each edge of this component has length at most \((2 \reach_i) / n\) or some would have uncontracted using a smaller reach.
    A shortest path between any two vertices of the component has length at most \((n-1) \cdot (2 \reach_i) / n < 2 \reach_i\).
\end{proof}

Along with each layer \(G_i\), we construct a Thorup-Zwick distance oracle (Section~\ref{sec:oracle}) with parameter~\(k := \lg n\).
Let \(S_i^j\), \(B_i^j(v)\), and \(B_i(v)\) denote the \(j\)th sample, \(j\)th bunch piece of \(v\), and bunch of \(v\), respectively, within layer \(G_i\).
By Lemma~\ref{lem:layers_size}, the oracles have total size~\(O(\log n \cdot m \log n \cdot n^{1 / \lg n}) = O(m \log^2 n)\),
and \(B_i(v) = O(\log^2 n)\) for each \(i, v\).
They can be constructed in~\(O(m \log^3 n)\) time total.

%% file: linear_cost_approx.tex
\section{A linear cost approximator}
\label{sec:linear_cost_approx}

In this section, we describe how to implicitly build and efficiently evaluate matrix-vector multiplications with an \(O(\log^3 n)\)-approximate linear cost approximator~\(\lcapprox \in \R^{((\cup_i V_i) \times (\cup_i V_i)) \times V}\),
i.e., the number of rows is \(|\cup_i V_i|^2\) and the number of columns is \(n\).
Most rows are empty.
To simplify the exposition, we will actually define three separate matrices \(\lcaggregate \in \R^{(\cup_i V_i) \times V}\),
\(\lcroute \in \R^{((\cup_i V_i) \times (\cup_i V_i)) \times (\cup_i V_i)}\), and \(\lccost \in \R^{((\cup_i V_i) \times (\cup_i V_i)) \times ((\cup_i V_i) \times (\cup_i V_i))}\) that represent \textbf{A}ggregating demands,
\textbf{R}outing flow, and estimating the \textbf{C}ost of the flow, respectively.
Approximator \(\lcapprox := \lccost \lcroute \lcaggregate\) is their product.
Each of these three matrices serves a limited purpose in a three-step process of obtaining a cost approximation.
Matrices \(\lcroute\) and \(\lccost\) are sparse and can be built and stored explicitly by our algorithm.
On the other hand, matrix \(\lcaggregate\) may be dense, so each multiplication with it will be done using a simple dynamic programming procedure.

\subsection{\(\lcaggregate\): Aggregating demands}
\label{sec:linear_cost_approx-aggregating}

Recall, each vertex of each layer graph is formed from the contraction of one or more edges from \(G\).
The flow modeled by our linear cost approximator sends the net amount of demand within each \(v' \in V_i\) to other vertices of \(G_i\).
We define the \EMPH{aggregate demand} of each vertex \(v' \in V_i\) to be \(\demand(v') := \sum_{v \in \contractset(v')} \demand(v)\);
our linear cost approximator is based on the cost of flow paths moving these aggregate demands between layer graph vertices.
Matrix \(\lcaggregate\) computes these aggregate demands so we may subsequently understand the cost of the flow.
For any \(i\), for any \(v' \in V_i\), and for any input vertex \(v \in V\),
\[\lcaggregate(v', v) := [v \in \contractset(v')]\]
where \([Q]\) denotes the \(0,1\)-indicator variable for proposition \(Q\).
For any demand vector \(\demand \in \R^V\), we have \((\lcaggregate \demand)(v') = \demand(v')\).

\begin{lemma}
    \label{lem:lcapprox1}
    Let \(\demand \in \R^V\) and \(\demand' \in \R^{\cup_i V_i}\).
    Vectors \(\lcaggregate \demand\) and \((\lcaggregate)^T \demand'\) can both be computed in \(O(m \log n)\) time.
\end{lemma}
\begin{proof}
    Consider any layer graph vertex \(v'\).
    If \(\contractset(v') = \Set{v}\) for some \(v \in V\), then \((\lcaggregate \demand)(v') = \demand(v)\).
    Otherwise, we have \(\demand(v') = \sum_{w \in \children(v')} \demand(w)\).
    We compute all entries \(\lcaggregate \demand(v)\) in \(O(m \log n)\) time by iterating through vertices in \emph{decreasing} order of layer graph index.

    Let \(v \in V\) be any input graph vertex, and let \(v'\) be the unique leaf such that \(\contractset(v') = \Set{v}\).
    We have \(((\lcaggregate)^T \demand')(v) = \sum_{w \mid v \in \contractset(w)} \demand'(w) = \sum_{w \in \ancestors(v')} \demand'(w)\).
    We compute \(\sum_{w' \in \ancestors(w)} \demand'(w')\) for all layer graph vertices \(w\) in \(O(m \log n)\) time by iterating through vertices in \emph{increasing} order of layer graph index.
    We then look up the values for the leaves in \(O(n)\) additional time.
\end{proof}

\subsection{\(\lcroute\): Routing flow}
\label{sec:linear_cost_approx-proportioning}

Matrix \(\lcroute\) is meant to model an oblivious routing of flows to satisfy the aggregate demands within each layer graph.
In Section~\ref{sec:approx_analysis-based_flow}, we show how the entries in \(\lcroute \lcaggregate \demand\) can be cleanly turned into a flow satisfying \(\demand\) itself.
The flows within each layer graph are determined by considering how much demand would be distributed to each potential cluster center in a flow based on a random tree embedding.
The hope is that opposing demands of nearby vertices see distribution to common targets, causing the demands to cancel.
In turn, the flows in subsequent layers end up routing only relatively light uncanceled portions of the original demands, keeping costs low.

We begin by computing a matrix \(\lcdistribute \in \R^{(\cup_i V_i) \times (\cup_i V_i)}\) that takes aggregate demands to their targets after \textbf{D}istribution.
Fix index~\(i\).
Consider any \(v \in V_i\).
Imagine continuously increasing a distance parameter~\(\lambda\) starting from~\(0\) and ending at~\(\reach_i\).
We wish to distribute the demand \(\demand(v)\) of \(v\) to members of \(B_i(v)\), giving precedence to those vertices in \(B_i(v)\) that are closer to \(v\).
The total fraction of demand distributed up to each moment \(\lambda\) should be equal to \(\lambda / \reach_i\).

Fix a moment~\(\lambda\).
There is a maximum index~\(j\) such that \(B_i^j(v)\) contains at least one vertex of distance at most~\(\lambda\) from \(v\).
As \(\lambda\) continues to increase, we will distribute the demand equally among exactly the vertices in \(B_i^j(v)\) at distance at most \(\lambda\) from \(v\).
This choice of equal distribution models each of those vertices being equally likely to be the center of the first ball of radius \(\lambda\) to contain \(v\).
As \(\lambda\) increases, the specific vertices within distance \(\lambda\) will change along with the index \(j\).

We now describe how to compute the total proportion of demand that should be distributed to each vertex in \(B_i(v)\).
Fix any \(j \in \Set{0, \dots k - 1 = \lg n - 1}\).
For any \(\lambda \geq 0\), let \(\bar{B}_i^j(v, \lambda) := \Set{w \in B_i^j(v) \mid \dist_i(v, w) \leq \lambda}\) denote those members of \(B_i^j(v)\) that are within distance \(\lambda\) of \(v\).
Let \(\lambda^{j+} := \min \Set{\dist(v, B_i^{j+1}(v)), \reach_i}\) (we define the distance to an empty set such as \(B_i^{k}(v)\) to be \(+\infty\)).
We sort the members of \(\bar{B}_i^j(v, \lambda^{j+})\) in increasing order of distance from \(v\) in \(G_i\) in \(O(|B_i^j(v)| \log n)\) time.
Let \(\langle w_1, w_2, \dots, w_r \rangle\) be this sorted sequence of vertices.
Finally, for each \(q \in \Set{1, \dots, r}\), we set
\[\lcdistribute(w_q, v) := \frac{\lambda^{j+} - \dist_i(v, w_r)}{r \reach_i} + \sum_{\ell = q}^{r - 1} \frac{\dist_i(v, w_{\ell+1}) - \dist_i(v, w_{\ell})}{\ell \reach_i}.\]
After sorting, the values \(\lcdistribute(w_q, v)\) for a particular \(i\), \(v\), and \(j\) can be computed in \(O(r) = O(|B_i^j(v)|)\) time as a running suffix sum.
All members of \(\lcdistribute\) not defined above are set to \(0\).

Observe for any vertex \(v \in (\cup_i V_i)\), \(\sum_{w \in (\cup_i V_i)} \lcdistribute(w, v) = 1\).
Each non-zero entry of \(\lcdistribute\) corresponds to the member of some bundle \(B(v)\), so there are at most \(O(m \log^2 n)\) of them.
Accordingly, we store \(\lcdistribute\) explicitly.
Accounting for the time needed to sort, \(\lcdistribute\) can be constructed in \(O(m \log^3 n)\) time total.

After computing \(\lcdistribute\), we are able to compute \(\lcroute\) itself.
For a given vertex \(v \in V_i\) with parent \(\parent(v) \in V_{i'}\) and \(w \in V_i\), we model routing the portion of \(v\)'s aggregate demand distributed to \(w\) to all \(w' \in V_{i'}\) proportionally to how much of \(\parent(v)\)'s demand should be distributed to the various \(w'\).
More concisely, we set
\[\lcroute((w, w'), v) := \lcdistribute(w', \parent(v)) \cdot \lcdistribute(w, v).\]
Finally, we need to route the aggregate demand of vertices in \(G_0\).
Let \(s \in V_0\) be arbitrarily chosen.
For each \(v, w \in V_0\), we set
\[\lcroute((w, s), v) := \lcdistribute(w, v).\]
Recall, each bundle \(B_i(v)\) has size at most \(O(\log^2 n)\).
Therefore, the total number of non-zero entries in \(\lcroute\) is at most \(O(m \log n) \cdot O(\log^2 n) \cdot O(\log^2 n) = O(m \log^5 n)\).
Again, we store the matrix explicitly.

\subsection{\(\lccost\): Estimating flow costs}
\label{sec:linear_cost_approx-charging}

We now compute the matrix \(\lccost\) whose job it is to estimate the costs of the flows described by \(\lcroute \lcaggregate \demand\).
As we shall see in the next section, these flows can follow relatively short paths so that the cost per unit flow sent within a layer graph \(G_i\) is \(O(\reach_i)\).
We define \(\lccost\) to be a diagonal matrix.
In order to keep it sparse, we only give it non-zero entries for pairs \((w, w')\) where row \((w, w')\) of \(\lcroute\) has at least one non-zero entry.
For each such \(w \in V_{i}\) and \(w' \in V_{i'}\), with \(\reach_{i'} = 2\reach_i\), we set
\[\lccost((w, w'), (w, w')) := 3 \reach_{i'},\]
and for each such \(w \in V_{i}\) and \(w' \in V_{i'}\), with \(i' > 2\reach_i\), we set
\[\lccost((w, w'), (w, w')) := 2 \reach_{i}.\]
(Note that \(i' = i\) for the case of \(w' = s\) as defined above.)
As before, the number of non-zero entries is \(O(m \log^5 n)\).

Recall, our \(O(\log^3 n)\)-approximate linear cost approximator \(\lcapprox := \lccost \lcroute \lcaggregate\).
We can perform matrix-vector multiplications with \(\lcapprox\) and its transpose by applying Lemma~\ref{lem:lcapprox1} when working with \(\lcaggregate\) and the standard multiplication algorithm when working with \(\lcroute\) and \(\lccost\).
\begin{lemma}
    \label{lem:lcapprox-apply}
    Let \(\demand \in \R^V\) and \(\demand' \in \R^{\cup_i V_i}\).
    Vectors \(\lcapprox \demand\) and \(\lcapprox^T \demand'\) can both be computed in \(O(m \log^5 n)\) time.
\end{lemma}

%% file: approx_analysis.tex
\section{Cost approximation analysis}
\label{sec:approx_analysis}

In this section, we establish the approximation ratio \(\approxratio = O(\log^3 n)\) of the linear cost approximator~\(\lcapprox\).
Fix any proper demand vector \(\demand \in \R^V\).
We define a flow~\(f \in \R^{\dartsof{E}}\) routing~\(\demand\) where \(\cost(f) \leq ||\lcapprox \demand||_1\).
Then, we prove \(||\lcapprox \demand||_1 \leq \approxratio \opt(\demand)\), giving a concrete expression for \(\approxratio\) in the process.

\subsection{A flow based on \(\lcapprox\)}
\label{sec:approx_analysis-based_flow}

For each pair of vertices \(u,v \in V\), let \(\shortestpath(u, v)\) denote the flow sending one unit from \(u\) to \(v\) along an arbitrary chosen \EMPH{canonical shortest path} in \(G\).
For each layer graph vertex \(w\), we pick an arbitrary representative \(\representative(w) \in \contractset(w)\).

We construct the flow~\(f\) as follows.
Initially, \(f = 0^{\dartsof{E}}\).
Then, for each non-zero \(\lcroute((w, w'), v)\), we add
\(\lcroute((w,w'), v) \cdot \demand(v) \cdot \shortestpath(\representative(w), \representative(w'))\) to \(f\).
In words, we send an \(\lcroute((w, w'), v)\) proportion of the \(\demand(v)\) units of flow demanded by \(v\) along the canonical shortest path between the representatives for \(w\) and \(w'\).

\begin{lemma}
    \label{lem:based_flow_routes}
    Flow \(f \in \R^{\dartsof{E}}\) as defined above routes \(\demand\).
\end{lemma}
\begin{proof}
    By construction, for any \(v \in V_i\) with \(i > 0\) and \(\parent(v) \in V_{i'}\), we have \(\sum_{w \in V_i} \lcdistribute(w, v) = 1\) and \(\sum_{w' in V_{i'}} \lcdistribute(w', \parent(v)) = 1\).
    Therefore, \(\sum_{w \in V_i} \sum_{w' \in V_{i'}} \lcroute((w, w'), v) = 1\), meaning the shortest path flows for \(v\) send a total of \(\demand(v)\) units
    from various \(w \in V_i\) to various \(w' \in V_{i'}\).
    In turn, \(\demand(\parent(v)) = \sum_{u \in \children(\parent(v))}\), so the amount taken from these various \(w' \in V_{i'}\) on behalf of \(\parent(v)\) is equal to the amount given on behalf of its children.

    The only targets \(w\) that do not have their flow carried further along by processing vertices \(v\) in lower index layers are those \(w \in V_0\).
    Each receives \(\sum_{v \in V_0} \lcdistribute(w, v) \demand(v)\) units of flow total, which are then carried to \(s\).
    Vertex \(s\) receives a net of \(\sum_{v \in V_0} \demand(v) = 0\) units from these paths.
    
    Finally, each vertex \(v \in G\) appears in exactly one layer graph leaf where it sends \(\demand(v)\) units out along one or more (possibly trivial) shortest path flows.
\end{proof}

\begin{lemma}
    \label{lem:flow_path_length}
    Let \(v,w \in V_i\) and \(w' \in V_{i'}\).
    If \(\lcroute((w, w'), v) \neq 0\) then
    \(\cost(\shortestpath(\representative(w), \representative(w'))) < 3 \reach_{i'}\).
\end{lemma}
\begin{proof}

    First, suppose \(i' = i = 0\).
    Every shortest path in \(G\) has length at most \(\reach_0 < 2\reach_{i'}\).

    Now, suppose \(i' < i\).
    Value \(\lcdistribute(w, v) > 0\) only if \(\dist_{i}(v, w) \leq \reach_{i} \leq \reach_{i'} / 2\), and \(\lcdistribute(w', \parent(v)) > 0\) only if \(\dist_{i'}(\parent(v), w') \leq \reach_{i'}\).

    The projection of any path in \(G_i\) to \(G\) includes some edges of \(E_i\) in addition to at most \(n - 1\) additional contracted edges of total length at most \((n - 1) \cdot \reach_i / n < \reach_{i'} / 2\).
    Similarly, projecting paths in \(G_{i'}\) to \(G\) increases the total length by at most \((n - 1) \cdot \reach_{i'} / n < \reach_{i'}\).
    By the triangle inequality, \(\cost(\shortestpath(\representative(w), \representative(w'))) < \dist_{i}(w, v) + \reach_{i'} / 2 + \dist_{i'}(\parent(v), w') + \reach_{i'} \leq 3\reach_{i'} \).
\end{proof}

\begin{lemma}
    \label{lem:based_flow_cost}
    We have \(\cost(f) \leq ||\lcapprox \demand||_1\).
\end{lemma}
\begin{proof}
    Fix any \(w \in V_i\) and \(w' \in V_{i'}\).
    We send a net of \(\sum_{v \in V_i} \lcroute((w, w'), v) \demand(v)\) units of flow along the canonical shortest path from \(\representative(w)\) to \(\representative(w')\).

    Suppose \(\reach_{i'} = 2\reach_i\).
    Even ignoring the potential for cancellation with other contributions to \(f\), Lemma~\ref{lem:flow_path_length} guarantees the total cost of the flow along this path is at most \(|3 \reach_{i'} \sum_{v \in V_i} \lcroute((w, w'), v) \demand(v)| = |(\lccost \lcroute \lcaggregate \demand)((w, w'))| = |(\lcapprox \demand)((w, w'))|\).

    Now, suppose \(i' > 2\reach_{i}\).
    By Lemma~\ref{lem:parent_structure}, \(\dist(w,w') \leq 2 \reach_i\), so the total cost of the flow along the path is at most \(|2 \reach_i \sum_{v \in V_i} \lcroute((w,w'), v) \demand(v)| = |\lcapprox \demand)((w, w'))|\).
\end{proof}

\subsection{Approximation ratio}
\label{sec:approx_analysis-ratio}

Our approximation ratio upper bound depends on the following lemma, loosely mirroring the fact that a random tree embedding of a graph distorts distances by only a small factor in expection.
The lemma essentially states that the closer two vertices sharing a layer graph are to one-another, the less they disagree on where their aggregate demand should be reassigned.
Let \(H_n = 1/1 + 1/2 + \dots + 1/n\) denote the \(n\)th harmonic number.

\begin{lemma}
\label{lem:distortion}
Let \(u,v \in V_i\) for some \(i\).
We have
\[\sum_{w \in V_i} |\lcdistribute(w, u) - \lcdistribute(w, v)| < \frac{8\dist_i(u, v) H_n \lg n}{\reach_i}.\]
\end{lemma}
\begin{proof}
   Recall in the construction of \(\lcdistribute\), we consider a continuously increasing \(\lambda \in [0, \reach_i]\).
   As \(\lambda\) increases at a rate of \(1\), we increase the value of at least one \(\lcdistribute(w, u)\) at a rate of \(1 / (\ell\reach_i)\) where \(\ell\) is the number of vertices in a particular set \(\bar{B}_i^j(u, \lambda)\).
   Only vertices \(w \in \bar{B}_i^j(u, \lambda) \subseteq B_i^j(u)\) see \(\lcdistribute(w, u)\) increase for this value of \(\lambda\).
   Let \(\close(u, \lambda)\) denote the set of vertices for which \(\lcdistribute(w, u)\) is increasing for parameter \(\lambda\), and define \(\close(v, \lambda)\) similarly.
   For a particular \(w \in V_i\), \(|\lcdistribute(w, u) - \lcdistribute(w, v)|\) is increasing at a rate of no more than
   \(|[w \in \close(u, \lambda)] / (|\close(u, \lambda)| \reach_i) - [w \in \close(v, \lambda)] / (|\close(v, \lambda)| \reach_i)|\).
   Suppose \(1 \leq |\close(u, \lambda)| \leq |\close(v, \lambda)|\).
   Then, the total rate of increase in difference among all \(w \in V_i\) for a particular \(\lambda\) is at most
   \begin{align*}
    \sum_{w \in V_i} \Abs{\frac{[w \in \close(u, \lambda)]}{\reach_i|\close(u, \lambda)|} - \frac{[w \in \close(v, \lambda)]}{\reach_i|\close(v, \lambda)|}}
    &= \frac{1}{\reach_i} \Paren{
        \begin{aligned}
            &\frac{|\close(u, \lambda) \setminus \close(v, \lambda)|}{|\close(u, \lambda)|} + \frac{|\close(v, \lambda) \setminus \close(u, \lambda)|}{|\close(v, \lambda)|}\\
            &\quad+ |\close(u, \lambda) \cap \close(v, \lambda)|\Paren{\frac{1}{|\close(u, \lambda)|} - \frac{1}{|\close(v, \lambda)|}}
        \end{aligned}}\\
    &= \frac{1}{\reach_i} \Paren{
        \begin{aligned}
            &\frac{|\close(u, \lambda) \setminus \close(v, \lambda)|}{|\close(u, \lambda)|} + \frac{|\close(v, \lambda) \setminus \close(u, \lambda)|}{|\close(v, \lambda)|}\\
            &\quad + \frac{|\close(u, \lambda) \cap \close(v, \lambda)|(|\close(v, \lambda)| - |\close(u, \lambda)|)}{|\close(u, \lambda)||\close(v, \lambda)|}
        \end{aligned}}\\
    &\leq \frac{1}{\reach_i} \Paren{
        \begin{aligned}
            &\frac{|\close(u, \lambda) \setminus \close(v, \lambda)|}{|\close(u, \lambda)|} + \frac{|\close(v, \lambda) \setminus \close(u, \lambda)|}{|\close(v, \lambda)|}\\
            &\quad + \frac{|\close(u, \lambda)||\close(v, \lambda) \setminus \close(u, \lambda)|}{|\close(u, \lambda)||\close(v, \lambda)|}
        \end{aligned}}\\
    &= \frac{1}{\reach_i} \Paren{\frac{|\close(u, \lambda) \setminus \close(v, \lambda)|}{|\close(u, \lambda)|} + \frac{2|\close(v, \lambda) \setminus \close(u, \lambda)|}{|\close(v, \lambda)|}}\\
    &\leq \frac{2}{\reach_i} \Paren{\frac{|\close(u, \lambda) \setminus \close(v, \lambda)|}{|\close(u, \lambda)|} + \frac{|\close(v, \lambda) \setminus \close(u, \lambda)|}{|\close(v, \lambda)|}}\\
    &= \frac{2}{\reach_i}\sum_{w \in V_i} \Paren{
        \begin{aligned}
            &\frac{[w \in (\close(u, \lambda) \setminus \close(v, \lambda))]}{|\close(u, \lambda)|} \\
            &\quad+ \frac{[w \in (\close(v, \lambda) \setminus \close(u, \lambda))]}{|\close(v, \lambda)|}
        \end{aligned}},
   \end{align*}
   and this same bound holds when \(|\close(v,\lambda)| < |\close(u, \lambda)|\).

   Considering all \(\lambda\), we see
   \[\sum_{w \in V_i} |\lcdistribute(w, u) - \lcdistribute(w, v)| \leq \frac{2}{\reach_i}\int_0^{\reach_i}\sum_{w \in V_i} \Paren{
    \begin{aligned}
        &\frac{[w \in (\close(u, \lambda) \setminus \close(v, \lambda))]}{|\close(u, \lambda)|} \\
        &\quad+ \frac{[w \in (\close(v, \lambda) \setminus \close(u, \lambda))]}{|\close(v, \lambda)|}
    \end{aligned}} d\lambda.\]

    Therefore, our goal is to, for each \(w \in V_i\), bound the measure of \(\lambda\) that puts \(w\) in exactly one of \(\close(u, \lambda)\) or \(\close(v, \lambda)\) and then divide by a number at most the size of that same set for the length of those \(\lambda\).

    Fix \(w\), and consider the set of \(\lambda\) such that \(w \in (\close(u, \lambda) \setminus \close(v, \lambda))\).
    We have two (not mutually exclusive) cases to consider.
    \begin{enumerate}
        \item
        \(\dist_i(u, w) \leq \lambda\) but \(\dist_i(v, w) > \lambda\):
        By the triangle inequality, \(\dist_i(v, w) - \dist_i(u, w) \leq \dist_i(u, v)\).
        Therefore, the range of \(\lambda\) for which this case can occur has measure at most \(\dist_i(u, v)\) as well.
        Recall, \(w \in \close(u, \lambda)\) implies \(w \in B_i^j(u)\) for some \(j\).
        Let \(\Seq{w_1, \dots, w_r}\) denote the vertices of \(B_i^j(u)\) sorted by increasing distance from \(u\), and let \(w = w_q\).
        We have \(\Set{w_1, \dots, w_q} \subseteq \close(u, \lambda)\), so \(|\close(u, \lambda)| \geq q\).

        \item
        \(\close(u, \lambda) \subseteq B_i^j(u)\) but \(\close(v, \lambda) \subseteq B_i^{j'}(v)\) for some \(j' \neq j\):
        This case occurs when (\(\dist_i(u, S_i^j) \leq \lambda\) and \(\dist_i(v, S_i^j) > \lambda\)) or (\(\dist_i(u, S_i^{j+1} > \lambda)\) and \(\dist_i(v, S_i^{j+1}) \leq \lambda\)).
        As before, the triangle inequality implies \(\dist_i(v, S_i^j) - \dist_i(u, S_i^j)\) and \(\dist_i(u, S_i^{j+1}) - \dist_i(v, S_i^{j+1})\) are both at most \(\dist_i(u, v)\).
        Therefore, the range of relevant \(\lambda\) for this case has measure at most \(2 \dist_i(u, v)\).
    \end{enumerate}

    There are \(k = \lg n\) different choices for \(j\) in either case.
    The total measure of \(\lambda\) resulting in case 2. is at most \(2 \dist_i(u, v) \lg n < \dist_i(u, v) H_n \lg n\).
    During those moments, \(\close(u, \lambda)\) and \(\close(v, \lambda)\) are completely disjoint, implying the integrand above is equal to \(2\).
    Summing across \(w \in V_i\), and noting each index \(q\) can be used at most once per choice of \(j\), we have
    \begin{align*}
        \sum_{w \in V_i} |\lcdistribute(w, u) - \lcdistribute(w, v)| 
        &\leq \frac{2}{\reach_i}\int_0^{\reach_i} \sum_{w \in V_i} \Paren{
        \begin{aligned}
            &\frac{[w \in (\close(u, \lambda) \setminus \close(v, \lambda))]}{|\close(u, \lambda)|} \\
            &\quad+ \frac{[w \in (\close(v, \lambda) \setminus \close(u, \lambda))]}{|\close(v, \lambda)|}
        \end{aligned}} d\lambda \\
        &< \frac{2}{\reach_i} \Paren{2 \dist_i(u, v) H_n \lg n + 2 \lg n \sum_{q = 1}^n \frac{\dist_i(u, v)}{q}} \\
        &= \frac{8\dist_i(u, v) H_n \lg n}{\reach_i}.
    \end{align*}
\end{proof}

We set \(\approxratio := 180 H_n \lg^2 n = O(\log^3 n)\).

\begin{lemma}
    \label{lem:ratio}
    We have \(||\lcapprox \demand||_1 \leq \approxratio \cdot \opt(\demand)\). 
\end{lemma}
\begin{proof}
    Let \(f^*\) be an optimal flow with \(\cost(f^*) = \opt(\demand)\).
    Standard flow theory implies there exists a decomposition of \(f^*\) into a linear combination \(a_1 f_1 + a_2 f_2 + \dots\) of unit path flows such that each \(a_j > 0\) and \(\cost(f^*) = \sum_j a_j \cost(f_j)\).
    We will charge each of the non-zero terms in \(\lcapprox \demand\) to one or more of these path flows and argue that each flow \(f_j\) is charged at most \(\approxratio a_j \cost(f_j)\).

    Fix \(f_j\).
    Let \(u \in V\) and \(v \in V\) be the source and sink of \(f_j\), respectively.
    We have \(\cost(f_j) \geq \dist(u, v)\) (in fact, equal).
    We consider charges based on \(a_j\) units in \(\demand(u)\) and \(-a_j\) units in \(\demand(v)\).
    These charges are handled in a few cases that are not necessarily exclusive.

    Consider any \((i', u')\) where \(\parent(u') \in V_{i'}\) and \(u \in \contractset(u')\),
    and \(\reach_{i'} \leq \dist(u, v)\).
    By construction of \(\lcroute\), \(\sum_{w \in (\cup_{i > i'} V_i)} \sum_{w' \in V_{i'}} \lcroute((w, w'), u') = 1\).
    Summing over all such pairs \((i', u')\), we see
    \begin{align*}
        \sum_{(i', u') \mid \parent(u') \in V_{i'} \wedge u \in \contractset(u') \wedge \reach_{i'} \leq \dist(u,v)} ~ \sum_{w \in (\cup_{i > i'} V_i)} \sum_{w' \in V_{i'}} 3\reach_{i'} \lcroute((w, w'), u')
        &= \sum_{i' \mid \reach_{i'} \leq \dist(u,v)} 3 \reach_{i'} \\
        &\leq 6 \dist(u,v) \\
        &< 2 \dist(u,v) H_n \lg^2 n.
    \end{align*}

    In words, the \(a_j\) units of demand from \(u\) contribute at most \(2 a_j \dist(u,v) H_n \lg^2 n\) total among various \(|(\lcapprox \demand)((w, w'))|\) where \(w' \in V_{i'}\) with \(\reach_{i'} \leq \dist(u, v)\).
    We charge \(2 a_j \dist(u,v) H_n \lg^2 n\) to \(f_j\) for these contributions.


    Now, consider any \((i', u')\) where \(\parent(u') \in V_{i'}\), \(u \in \contractset(u')\), and \(\reach_{i'} > \dist(u, v)\).
    Having \(\reach_{i'} > \dist(u, v)\) implies either there is a path from \(\parent(u')\) to a distinct contracted set of vertices in \(G_i\) containing \(v\) or \(v \in \contractset(\parent(u'))\) as well.
    Let \(v'\) such that \(\parent(v') \in V_{i'}\) with \(v \in \contractset(v')\).

    Suppose either \(u'\) or \(v'\) is in some \(V_i\) with \(\reach_i < \reach_{i'} / 2\).
    Lemma~\ref{lem:parent_structure}'s guarantee of long incident edges for \(\parent(u')\) and \(\parent(v')\) implies \(\parent(v') = \parent(u')\), and no other vertex of \(V_{i'}\) lies within \(\reach_{i'}\) of them.
    Therefore,
    \(\lcdistribute(\parent(v'), v) = \lcdistribute(\parent(v'), u') = 1\), and
    \begin{align*}
    \sum_{w \in V_i} \sum_{w' \in V_{i'}} |&\lcroute((w, w'), u') - \lcroute((w, w'), v')| \\
     &= \sum_{w \in V_i} \sum_{w' \in V_{i'}} |\lcdistribute(w', \parent(u')) \lcdistribute(w, u') - \lcdistribute(w', \parent(v'))\lcdistribute(w, v')|  \\
    &= \sum_{w \in V_i} |\lcdistribute(w, u') - \lcdistribute(w, v')| \\
    &\leq \frac{8\delta(u,v)H_n \lg n}{\reach_{i}}\\
    &\leq \frac{4\delta(u,v)H_n \lg^2 n}{\reach_{i}}.
    \end{align*}
    
    Therefore, all but \(4a_j\delta(u,v)H_n \lg^2 n / \reach_{i}\) units of the \(a_j\) demand from \(u\) contributing to various \(|(\lcapprox \demand)(w, w')|\) of this sort are canceled by opposite demand from \(v\).
    Each unit is multiplied by \(2 \reach_{i}\), so we charge \(8a_j\delta(u,v)H_n\lg n\) for these contributions.

    This subcase can only occur once, because for smaller choices of \(i'\), vertices \(u'\), \(v'\), and their parents have identical aggregate demand distributions, implying \emph{all} \(a_j\) units of demand from \(u\) are cancelled by opposite demand from \(v\).

    If the above subcase does not occur, then both \(\parent(u')\) and \(\parent(v')\) belong to the specific common \(V_i\) with \(\reach_i = \reach_{i'} / 2\).
    By Lemma~\ref{lem:distortion} and the fact that \(|ab - cd| \leq |a - c| + |b - d|\) for \(a,b,c,d \in [0, 1]\),
    \begin{align*}
    \sum_{w \in V_i} \sum_{w' \in V_{i'}} |&\lcroute((w, w'), u') - \lcroute((w, w'), v')| \\
    &= \sum_{w \in V_i} \sum_{w' \in V_{i'}} |\lcdistribute(w', \parent(u')) \lcdistribute(w, u') - \lcdistribute(w', \parent(v'))\lcdistribute((w, v')| \\
    &\leq \sum_{w \in V_i} \sum_{w' \in V_{i'}}  (|\lcdistribute(w',\parent(u'))- \lcdistribute(w',\parent(v'))| + |\lcdistribute(w, u') - \lcdistribute(w, v')|) \\
    &\leq \frac{8\delta(u,v)H_n \lg n}{\reach_{i'}} + \frac{8\delta(u,v)H_n \lg n}{\reach_{i'} / 2}\\
    &\leq \frac{24\delta(u,v)H_n \lg n}{\reach_{i'}}.
    \end{align*}
    
    Therefore, all but \(24a_j\delta(u,v)H_n \lg n / \reach_{i'}\) units of the \(a_j\) demand from \(u\) contributing to various \(|(\lcapprox \demand)(w, w')|\) of this sort are canceled by opposite demand from \(v\).
    Each unit is multiplied by \(3 \reach_{i'}\), so we charge \(72a_j\delta(u,v)H_n\lg n\) for these contributions.

    The contractions used in building the layer graphs guarantees there is at most one choice of \(i'\) where \(\reach_{i'} \geq n \cdot \dist(u, v)\) while \(u' \neq v'\).
    For smaller choices of \(i'\), vertices \(u'\), \(v'\), and their parents have identical aggregate demand distributions, implying \emph{all} \(a_j\) units of demand from \(u\) are cancelled by opposite demand from \(v\).
    There are at most \(\lg n\) different values \(i'\) with large reach that result in any non-zero charge, bringing this set of charges to a total of \(72a_j\delta(u,v)H_n\lg^2 n\).

    Finally, we consider \(u' \in V_0\), \(v' \in V_0\), \(u \in \contractset(u')\), and \(v \in \contractset(v')\).
    We have
    \[\sum_{w \in V_0} |\lcdistribute(w, u') - \lcdistribute(w, v')| \leq \frac{8\delta(u,v)H_n\lg n}{\reach_0}.\]
    Similar to before, all but \(8a_j\delta(u,v)H_n \lg n / \reach_0\) units of the \(a_j\) demand from \(u\) are canceled by opposite demand from \(v\).
    Multiplying by \(2 \reach_0\), we charge \(16a_j\delta(u,v)H_n \lg n \leq 8a_j\delta(u,v)H_n \lg^2 n\).

    Summing over all types of charges and then doubling the value for the contributions of \(v\) to various \(|(\lcapprox \demand)(w, w')|\), we get a total charge of \(180a_j\delta(u,v)H_n\lg^2 n \leq \approxratio \cdot a_j\cost(f_j)\) to \(f_i\).
    All units of demand for all \(u,v \in V\) are considered throughout these charges, so we have charged at least \(||\lcapprox b||_1\) in total.
    On the other hand, we charge at most \(\sum_j \approxratio \cdot a_j \cost(f_j) = \approxratio \cdot \cost(f^*) = \approxratio \cdot \opt(\demand)\) in total.
\end{proof}

%% file: scheme.tex
\section{Approximation scheme}
\label{sec:scheme}

We are now ready to present our main theorem.
\begin{theorem}
    There exists a deterministic algorithm that given an undirected graph \(G = (V, E)\) over \(n\) vertices and \(m\) edges, positive edge costs \(\cost \in \R_{> 0}^{E}\), a proper set of demands \(\demand \in \R^V\), and a parameter \(\eps > 0\) computes a flow \(f\) routing \(\demand\) with \(\cost(f) \leq (1 + \eps)\opt(d)\) in \(O(\eps^{-2}m \log^{O(1)} n)\) time.
\end{theorem}
\begin{proof}
    We begin by constructing the layer graphs as presented in Section~\ref{sec:layers} along with the Thorup-Zwick distance oracles for each in \(O(m \log^3 n)\) time total.
    We construct sparse representations of the matrices \(\lcroute\) and \(\lccost\) in \(O(m \log^5 n)\) time so they can be used in matrix-vector multiplications with the \(O(\log^3 n)\)-approximate linear cost approximator \(\lcapprox\) and its transpose.
    Each matrix-vector multiplication takes \(O(m \log^5 n)\) time.
    We perform \(O(\eps^{-2} \log^{O(1)} n)\) multiplications with \(\lcapprox\) according to the boosting framework of Zuzic~\cite[Corollaries 12 and 16]{z-sbft-23} to find an infeasible flow \(f'\) of cost \(\cost(f') \leq (1 + \eps/2)\opt(\demand)\) where \(\opt(\demand - \incidence_G f') \leq \opt(\demand) / n^2\).
    Finally, we add in an \(n\)-approximate solution that routes \(\demand - \incidence_G f'\) as in~\cite{s-gpumf-17} in \(O(m \log n)\) additional time.
    The total cost of the flow returned is \((1 + \eps / 2 + 1/n) \opt(\demand) \leq (1 + \eps)\opt(\demand)\).
\end{proof}